\documentclass{article}

\usepackage{authblk}
\usepackage{url}
\usepackage[multiple]{footmisc} 
\usepackage{graphicx}      
\usepackage{amsthm}	
\usepackage{amsmath}     
\usepackage{amsfonts}	 
\usepackage{amssymb}	 
\usepackage{xcolor}	 	 
\usepackage[algo2e]{algorithm2e} 
\usepackage{tikz}		 
\usetikzlibrary{arrows, automata}
\usepackage{subcaption}	
\usepackage{mathtools}	
\usepackage{multirow}     
\usepackage[width=.75\textwidth]{caption} 

\DeclarePairedDelimiter{\ceil}{\lceil}{\rceil}

\newcommand{\dom}[1]{Dom\left(#1\right)}
\newcommand{\cl}{\mathnormal{g}}
\newcommand{\tcl}{\tilde{\cl}}

\newcommand{\bits}{bits}
\newcommand{\bitsf}[1]{\bits\left(#1\right)}

\newcommand{\clf}[1]{\cl\left(#1\right)}

\newcommand{\tclf}[1]{\tcl\left(#1\right)}

\newcommand{\R}{\mathbb{R}}
\newcommand{\N}{\mathbb{N}}

\newcommand{\Nnn}{\mathbb{Z}_{\geq 0}}

\newtheorem{problem}{Problem}[section]
\newtheorem{theorem}{Theorem}[section]
\SetKwRepeat{Until}{do}{while}

\begin{document}

\title{Optimal Symbolic Controllers Determinization for BDD storage.} 

\author{Ivan S. Zapreev, Cees Verdier, Manuel Mazo Jr.} 

\affil{Center for Systems and Control, Technical University of Delft, The Netherlands (e-mails: I.Zapreev@tudelft.nl, M.Mazo@tudelft.nl, C.F.Verdier@tudelft.nl)}

\maketitle

\begin{center}
\thanks{Supported by STW-EW as a part of the CADUSY project \#13852.}

\thanks{The short version of this article has been accepted to ADHS'2018.}
\end{center}

\begin{abstract}
Controller synthesis techniques based on symbolic abstractions appeal by producing correct-by-design controllers, under intricate behavioural constraints. Yet, being relations between abstract states and inputs, such controllers are immense in size, which makes them futile for embedded platforms. Control-synthesis tools such as PESSOA, SCOTS, and CoSyMA tackle the problem by storing controllers as binary decision diagrams (BDDs). However, due to redundantly keeping multiple inputs per-state, the resulting controllers are still too large. In this work, we first show that choosing an optimal controller determinization is an NP-complete problem. Further, we consider the previously known controller determinization technique and discuss its weaknesses. We suggest several new approaches to the problem, based on greedy algorithms, symbolic regression, and (muli-terminal) BDDs. Finally, we empirically compare the techniques and show that some of the new algorithms can produce up to $\approx 85$\% smaller controllers than those obtained with the previous technique.
\end{abstract}


\section{Introduction}
	Controller synthesis techniques based on symbolic models, such as e.g.~\cite{Tabuada_2009, Rungger_EtAl_2013, Liu_EtAl_2013}, are becoming increasingly popular. One of the key advantages of these techniques is that they allow for synthesising correct-by-construction controllers of general nonlinear systems under intricate behavioural requirements. However, the downside of the synthesised controllers is their size as, in essence, they are huge tables mapping abstract state-space elements into input-signal values. Even for toy examples, the produced controllers can reach a size of several megabytes. In real-life applications however, they can be several orders of magnitude larger. The latter prohibits them from being used on embedded micro-controllers which typically have very limited memory resources. In general, this state-space explosion is the consequence of:
(1) the number of abstract system states and inputs which are exponential in the number of dimensions and inverse-polynomial in the discretisation values; and (2) storing multiple valid input signals per abstract state.

There are numerous tools, implementing or incorporating control synthesis, such as PESSOA, SCOTS, CoSyMA, LTLMoP, TuLiP, see \cite{Mazo_EtAl_2010}, \cite{Rungger_EtAl_2016}, \cite{Mouelhi_EtAl_2013},  \cite{Finucane_EtAl_2010}, and \cite{Wongpiromsarn_EtAl_2011} correspondingly. Internally, they either use an explicit control law representation in a table form or employ Reduced Ordered Binary Decision Diagrams, introduced by \cite{Bryant_1986} and called RO-BDDs or simply BDDs, in an attempt to optimise the memory needed to store the synthesised control law. RO-BDDs are canonical, efficiently manipulable, and in many cases allow for compact data representation. However, their size is strongly dependent on the variables' ordering and the problem of finding an optimal one is known to be NP-complete, as shown by \cite{Bollig_1996}. To fight that issue, tools such as SCOTS and Pessoa use the state of the art RO-BDD library CUDD, see \cite{Somenzi_2015}, which implements numerous efficient variable ordering optimisation heuristics. Yet, even when using BDDs, controllers synthesised for practical applications can easily reach hundreds of megabytes.

To our knowledge, there have been just a few attempts made to find compact but practical representations of (symbolically produced) control laws. Except for using BDDs, we are only aware of another two approaches. The first one, suggested by \cite{Staudt_1998}, uses piece-wise linear functions, also known as linear in segments (LIS) functions, to approximate control functions of the form $\cl : \R \rightarrow \R$. The approximation is considered for scalar control functions of one argument only. The main motivation for LIS is to reduce the memory footprint of implementing controllers at the cost of some on-line computations, which nonetheless are fast to perform.
However, this approach does not directly scale to multiple dimensions or allows to resolve multiple-input's non-determinism.  

Another technique to reduce the control-law size, we shall refer to as \texttt{LA} (Local Algorithm), was proposed by \cite{Girard_2012}. It borrows ideas from algebraic decision diagrams (ADDs), see \cite{Bahar_EtAl_1993}, for compact function representation and exploits the non-determinism inherent to safety controllers. The considered controllers are multi-valued maps $\cl : \R^{n} \rightrightarrows \N$. The suggested approach attempts to optimise the controller size determinizing the control law by choosing one of the possible control signals for each of the state-space points. In the selection of such unique inputs, \texttt{LA} maximizes the size of state-space neighbourhoods employing the same input with the expected outcome of minimizing an ADD representation of the resulting control function.
However, the minimality of the ADD representation cannot be guaranteed in general by this approach, which leads us to investigate if better compression approaches may be viable.

In this paper, we first prove that the problem of choosing a size-optimal controller determinization is NP-complete. We do that assuming the BDD controller representation, but the result can be easily generalised. Next, we suggest two new determinization approaches : \texttt{GA} (Global Algorithm) - based on a greedy algorithm for the minimum set-cover selection problem, see \cite{Karp_1972, Chvatal_1979}; \texttt{SR} - a hybrid of ADD-based and symbolic regression techniques, powered by genetic programming, see \cite{Koza1994, Willis_EtAl_1997}. \texttt{GA} attempts to minimise the BDD size by maximising the number of controller states having the same input signal. It differs from \texttt{LA} in that, when choosing a common input for a set of states, it looks at the state-space globally, without considering the actual state positions. \texttt{SR} (Symbolic Regression) aims at bridging the intrinsic limitations of \texttt{LA} and \texttt{GA} by using ``arbitrary" (polynomial and sigmoid in our case) functions as controller representations. This way we realise the Kolmogorov's~\cite{Li_EtAl_2008} view on data compression\footnote{Instead of storing the control law as an explicit map, we search for a symbolic function that for a given state computes the input value.}. Further, we combine \texttt{LA} and \texttt{GA} into a hybrid approach called \texttt{LGA} (Local-Global Algorithm). The idea here is that the determinization is done as in \texttt{LA} but, if multiple common inputs are possible, the preference is given to the one suggested by \texttt{GA}. In addition, we consider \texttt{B}-prefixed version of \texttt{LGA} (\texttt{BLGA}) which attempts for a better compression by using BDD variable reordering to produce abstract state indexes.

We perform an empirical evaluation on a number of examples from the literature. 
Our results show that compression-wise\footnote{Up to the found optimal BDD variable reordering.} there is no absolute best approach. However, \texttt{LGA} seems, on most cases, to be providing the best compression. The \texttt{SR} approach, while only providing better compressions in few examples, may be most promising when looking at actual embedded deployments, if it could be pushed to remove any use of BDDs, and their overhead on actual implementations.

\section{Preliminaries} \label{sec:prelim}

\subsection{Minimum set cover} \label{sub:sec:prelim:min_set_cov}
The minimum set cover problem (\texttt{MSC}) is formulated as: 
\begin{problem}[\texttt{MSC}]
Given a set $X$ and a cover $\left\{S_{j}\right\}_{j \in I}$, i.e. $X\subseteq\bigcup_{j \in I}S_j$, where $\vert X \vert, \vert I \vert < \infty$, find the smallest subcover $I^* \subseteq I: X \subseteq \bigcup_{j \in I^*} S_{j}$.
\end{problem}
Both, the decision and selection versions of (\texttt{MSC}, are known to be NP-complete. The first approximate poly-nomial-time solution for \texttt{MSC} was given by \cite{Karp_1972}. Later, \cite{Chvatal_1979} suggested an approximate poly-nomial-time solution for the generalized {\em ``minimum set weight cover problem''} (\texttt{MWSC}); which extends \texttt{MSC} by that each set $S_{k}$ is assigned a weight $s_{k} \geq 0$ and the question is to find the smallest sub-cover with the minimum total weight.  According to \cite{Cormen_EtAl_2001}, the Chv\'{a}tal's algorithm time complexity is: $O\left(\vert I \vert \cdot \vert X \vert \cdot min\left(\vert I \vert, \vert X \vert\right)\right)$.

\subsection{Symbolic regression} \label{sub:sec:prelim:symb_reg}
	Symbolic regression is a type of regression analysis that searches for analytical expressions best fitting a given dataset of numerical data, both in terms of accuracy and simplicity. 
	We apply this technique in order to find the smallest analytical expressions best fitting symbolic-model-based control-law functions, ensuring for the smallest control law representation.
	One of the most popular means for symbolic regression is genetic programming, see~\cite{Koza_EtAl_1992} (GP). In this work, similar to~\cite{Whigham_EtAl_1995}, we employ grammar guided genetic programming algorithms (GGGP) to find multi-dimensional analytical expressions fitting the controller's data. In fact, the genetic process follows~\cite{Verdier_EtAl_2017} except for that the real-value parameter tuning is done with CMA-ES~\cite{Hansen_EtAl_2001}. To speed up the CMA-ES procedure, we use sep-CMA-ES which has a linear time and space complexity \cite{Ros_2008}.

\subsection{Binary Decision Diagrams} \label{sub:sec:prelim:bdds}
	\emph{Binary Decision Diagrams} (BDDs), represented with rooted directed acyclic graphs were introduced by \cite{Bryant_1986}, as a compact representation for boolean functions $F : \left\{0,1\right\}^{n}\rightarrow \left\{0,1\right\}$. Given $F$ with a list of arguments $\left\{v_{i}\right\}^{n}_{i=0}$, also called BDD variables or just variables, the BDD of $F$ results from the Shannon expansion thereof. The order of arguments in the signature of $F$ has clearly no impact on $F$ itself, but it has a drastic impact on the size of the resulting BDD. Finding a size-optimal BDD variable ordering was shown, in \cite{Bollig_1996}, to be NP-complete. Yet, there are multiple polynomial heuristics, \cite{Scholl_EtAl_1999}, that can find a semi-optimal variable ordering. One of the most popular thereof is sifting, \cite{Rudell_1993}, and its variants. Given a fixed variable order, each BDD has a canonical minimum-size representation, called Reduced Ordered BDD (RO-BDD).
Assuming the bottom-up BDD traversal, an RO-BDD can be obtained by the following poynomial-time algorithm, for more details see Section~$4.2$ of~\cite{Bryant_1986}:
	\begin{enumerate}
		\item Combine terminal nodes with equal values\;
		\item Eliminate nodes with equivalent\footnote{``Equivalent'' means: Representing the same binary function.} children\;
		\item Combine nodes with pairwise equivalent children\;
	\end{enumerate}
	
\emph{Multi Terminal BDDs} (MTBDDs) extend BDDs in that tree's terminal nodes allow for arbitrary labels, thus useful to encode functions of the form $F : \left\{0,1\right\}^{n}\rightarrow U$, with $|U|<\infty$. The BDD reduction algorithm is naturally extendable towards MTBDD which thus have the canonical RO-MTBDD form. For an (MT)BDD $M$, we define $R\left(.\right)$ as a reduction function producing the RO-(MT)BDD $R\left(M\right)$.  Algebraic Decision Diagrams (ADDs), introduced by \cite{Bahar_EtAl_1993}, are a synonym of MTBDDs. The current state of the art implementation for RO-(MT)BDDs is provided by the CUDD package~\cite{Somenzi_2015}. 

\subsection{SCOTS v2.0} \label{sub:sec:prelim:scots}
	SCOTS is an open source software that implements construction of symbolic models, also known as discrete abstractions, of possibly perturbed, nonlinear control systems. The tool natively supports invariance and reachability specifications as well as several control synthesis algorithms. The control laws can be stored in BDD. 
SCOTS comes in a form of a header-only C++ library that can be easily included in any C/C++ code but also has a MATLAB interface. We base our algorithms on the interfaces provided by the UniformGrid and SymbolicSet classes of the tool.
	
\section{Problem statement} \label{sec:problem}
Consider a (possibly non-linear) discrete time control system of the form:
$$
x(k+1)=f(x(k),u(k)),\;\, x(k)\in\mathcal{X}\subseteq\R^n, u(k)\in\mathcal{U}\subseteq\R^m.
$$
Symbolic approaches, see e.g.~\cite{Tabuada_2009}, automatically synthesize controllers in the form of discrete state transition systems. Furthermore, the resulting controllers can often be reduced to a look-up table, see~\cite{Reissig_EtAl_2016}, prescribing for each point of the state-space a set of applicable inputs guaranteeing that the control specification is satisfied. Such synthesized controllers usually take the form of the combination of a (finite) set-valued map $\cl : \mathcal{S} \rightrightarrows \mathcal{V}$, and quantization maps $\texttt{q}_x:\mathcal{X}\to\mathcal{S}$, $\texttt{q}_u:\mathcal{U}\to\mathcal{V}$ reducing the originally infinite state and input sets to finite sets (usually defining a grid), i.e. $\mathcal{S}\subset\mathcal{X}$, $\mathcal{S}\subset\mathcal{X}$, $|\mathcal{S}|<\infty$, $|\mathcal{V}|<\infty$. Moreover, the usual approach is to quantize each dimension of $\mathcal{X}$ and $\mathcal{U}$ independently, i.e. $\texttt{q}_x(x)=(\texttt{q}^1_x(x_1),\ldots,\texttt{q}^n_x(x_n))$, $\texttt{q}_u(u)=(\texttt{q}^1_u(u_1),\ldots,\texttt{q}^n_u(u_n))$, where each of the $\texttt{q}^i_{x}:\mathcal{X}_i\subset \R \to \mathcal{S}_i$, such that $\mathcal{S}=\mathcal{S}_1\times\ldots\times\mathcal{S}_n$, and similarly for the input quantizer. This results in controller implementations selecting at each time step $u(k)\in\clf{\texttt{q}_x(x(k))}$, see for details of such controllers~\cite{Reissig_EtAl_2016}. Most often, the controllers synthesized do not provide a valid input for some subset $\mathcal{S}_\emptyset\subset\mathcal{S}$. We define the set $\mathcal{S}_c:=\mathcal{S}\setminus\mathcal{S}_\emptyset$. We may assume that there is some element $\texttt{nc}\in\mathcal{V}$ denoting a ``no-input", and thus we can define $\mathcal{S}_\emptyset:=\cl^{-1}(\texttt{nc})$.   

A symbolic controller $\cl \subseteq \mathcal{S} \times \mathcal{V}$, by indexing the countable sets $\mathcal{S}_i$ and $\mathcal{V}_i$, can alternatively be interpreted as a relation $\cl \subset \Nnn \times \Nnn$. Consider $\mathcal{B} := \left\{0,1\right\}$, and let us define a fixed-length base-$2$ bit encoding for non-negative integers $\bits:\mathcal{K} \to \mathcal{B}^{b}$ for some $\mathcal{K} \subset  \Nnn$, $\vert\mathcal{K}\vert < \infty$, and $b := \ceil*{log_{2}\left(max\left(\mathcal{K}\right)\right)}$. For  $k = \left(k_{1}, k_{2}\right) \in \cl \subset \Nnn \times \Nnn$, mapping the bit vector $\left(\bitsf{k_{1}}, \bitsf{k_{2}}\right)$ to a boolean $1$  defines a BDD encoding of $\cl$. Similarly, one can construct an MTBDD encoding of $\cl$ by mapping $\bitsf{k_{1}}$ to $k_{2}$.

Relating elements of $\mathcal{S}$ or $\mathcal{V}$ with $\Nnn$ can be done with an indexing function, typically defined as:
	\begin{align}
		f_{b}\left(k_{a}, \ldots, k_{b}\right) := \sum^{b}_{i=a}k_{i}\cdot\left(\prod^{i-1}_{j=a}2^{\vert\bitsf{N_{j}}\vert}\right) \label{eq:cell_enc_cudd} \mbox{, or} \\
		f_{s}\left(k_{a}, \ldots, k_{b}\right) := \sum^{b}_{i=a}k_{i}\cdot\left(\prod^{i-1}_{j=a}N_{j}\right) \label{eq:cell_enc_scots}
	\end{align}
Here, $N_{j}$$:=$$\vert\mathcal{S}_j\vert$ for $j$$\in$$\overline{1,n}$, and $N_{j}$$:=$$\vert\mathcal{V}_j\vert$ for $j$$\in$$\overline{n\!+\!1,n\!+\!m}$; $\vert\bitsf{N_{j}}\vert$ is the data-type size needed to enumerate intervals in $j$. Equations~\ref{eq:cell_enc_scots}~and~\ref{eq:cell_enc_cudd} are both used in SCOTSv2.0. The former is employed in its interface classes (\texttt{UniformGrid} and \texttt{SymbolicSet}), as it delivers smaller indexes. The latter is used for BDD encoding as it avoids bit sharing between distinct dimension interval indices.
	
	In the present we consider the following minimisation problem aimed at finding the smallest controller determinization of a given controller $\cl$:
\begin{problem}[\texttt{OD}]
	Find the best determinization $\cl^*$ of a controller $\cl$ optimizing: $\cl^*=argmin_{\tcl  \in \mathcal{F}}\left\vert enc\left(\tcl\right)\right\vert$, where 
	\vspace{-0.2cm}
	\begin{eqnarray*}
		\mathcal{F}:&=&\left\{ \tcl : \Nnn \rightarrow \Nnn \vert \right. \\
		&& \left.\forall s \in \dom{\cl}: \left(\left(\tclf{s} \in \clf{s}\right) \wedge \left(\vert \tclf{s} \vert = 1\right)\right)\right\},
	\end{eqnarray*} 
	$enc\left(.\right)$ encodes controllers into RO-(MT)BDDs, and $\vert.\vert$ provides the (MT)BDD size. 
\end{problem}

In theoretical derivations, as in \cite{Kwiatkowska_2006}, we define $\vert.\vert$ to be the number of (MT)BDD nodes. In practice, $\vert.\vert$ is the number of bits used to store the (MT)BDD by the CUDD package in the best-found, variable reordering.

\section{\texttt{LA} on MTBDDs} \label{sec:la_weak}
	\cite{Girard_2012} suggests a controller-size minimisation technique, which we call \texttt{LA}, that uses ideas from MTBDDs to represent the controller function in the form of a binary tree. The approach does dimension-wise binary splitting of the controller's state-space bounding box. The areas with no-inputs are considered to allow for any input. For the areas with common inputs possible a single input is selected non-deterministically. A branch in the tree represents a state-space area with all states having common inputs (stored in terminal nodes). The determinization aims at choosing single inputs in a way minimising the depth of the tree branches. The latter is equivalent to reductions as in steps (1) and (2) of the RO-BDD construction (c.f. Section~\ref{sub:sec:prelim:bdds}), but not (3). 
	\cite{Girard_2012} showed that \texttt{LA} can lead to drastic size reductions, e.g., for ``the simple thermal model of a two-room building'' example the original controller required $1.000.000$ data units, whereas in the tree format it went down to $27$. Yet, in its original form this approach: \emph{(i)} does not preserve the controller's domain -- neglecting basic data of safe initial states; \emph{(ii)} employs a fixed state-space splitting algorithm -- not using controller's structural features;  \emph{(iii)} uses simple binary trees which are less efficient than MTBDDs, due to the latter compression abilities by variable reordering and their canonical reduced form. This motivates extending the approach towards MTBDDs.
	
	\texttt{LA} can be adapted to quantised state-spaces, since: 
	\begin{enumerate}
	\item[\emph{(i)}] For dimension $i$$\in$$\overline{1,n}$ and $s_{i}$$\in$$\mathcal{S}_{i}$, the bit sequence $\bitsf{s_{i}}$, defines a binary-tree path to $s_{i}$ in $\mathcal{S}_{i}$. 
	\item[\emph{(ii)}] For $s$$\in$$\mathcal{S}$, the alternating bit sequence obtained from $\bitsf{s_{1}},\ldots,\bitsf{s_{n}}$ defines a binary-tree path to $s$ in $\mathcal{S}$. 
	\end{enumerate}
	The latter, using bounded-length bit sequences as in Section~\ref{sec:problem}, allows to encode the \texttt{LA}'s binary tree as an MTBDD. The size reductions obtained for the original \texttt{LA} are then a subset of those we get using MTBDDs\footnote{Even with the original variable ordering.}, as we can: \emph{(i)} obtain RO-MTBDDs, utilising all the reduction steps \emph{(ii)} find a more efficient variable ordering. Let us now show that, however good, \texttt{LA} does not allow to utilise the full power of the MTBDD reductions due to its pure spacial orientation.
	
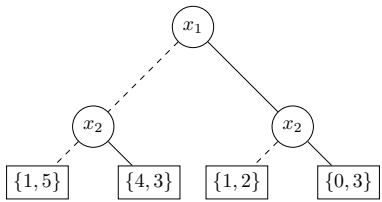
\begin{figure}
\centering
\scalebox{0.75}{
\begin{tikzpicture}[>=stealth, node distance=2.5cm]
	\node (A) [style={circle, draw}] {$x_{1}$};
	\node (B) [below left of=A, style={circle, draw}] {$x_{2}$};
	\node (C) [below right of=A, style={circle, draw}] {$x_{2}$};
	\node (D) [below left of=B, node distance=1.4cm, style={rectangle, draw}] {$\left\{1, 5\right\}$};
	\node (E) [below right of=B, node distance=1.4cm, style={rectangle, draw}] {$\left\{4, 3\right\}$};
	\node (F) [below left of=C, node distance=1.4cm, style={rectangle, draw}] {$\left\{1, 2\right\}$};
	\node (G) [below right of=C, node distance=1.4cm, style={rectangle, draw}] {$\left\{0, 3\right\}$};
	
	\draw[dashed](A) to node {}(B);
	\draw[](A) to node {}(C);
	\draw[dashed](B) to node {}(D);
	\draw[](B) to node {}(E);
	\draw[dashed](C) to node {}(F);
	\draw[](C) to node {}(G);
\end{tikzpicture}
}
\caption{An example MTBDD}\label{fig:mtbdd_ex_01}
\end{figure}

	Consider an MTBDD encoding of some \texttt{LA}'s binary tree, in its original variable ordering, see Figure~\ref{fig:mtbdd_ex_01}. \texttt{LA} traverses an MTBDD trying to find common inputs, stored in terminal nodes, for all of its sub-trees. A sub-tree with a common input can then be trivially reduced to a single terminal node. In this case however, there are no non-trivial sub-trees with common inputs, so \texttt{LA} has to \emph{non-deterministically} choose one (arbitrary) input value per terminal node. This results in $16$ possible determinization variants, most of which would not be reducible, see e.g. Figure~\ref{fig:mtbdd_ex_01_bad}, but a few would allow for reductions; the best one is in Figure~\ref{fig:mtbdd_ex_01_good}.

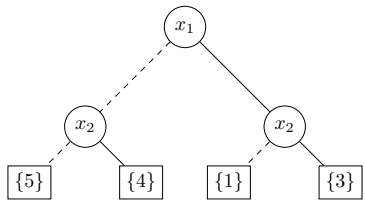
\begin{figure}
\centering
\scalebox{0.75}{
\begin{tikzpicture}[>=stealth, node distance=2.5cm]
	\node (A) [style={circle, draw}] {$x_{1}$};
	\node (B) [below left of=A, style={circle, draw}] {$x_{2}$};
	\node (C) [below right of=A, style={circle, draw}] {$x_{2}$};
	\node (D) [below left of=B, node distance=1.4cm, style={rectangle, draw}] {$\left\{5\right\}$};
	\node (E) [below right of=B, node distance=1.4cm, style={rectangle, draw}] {$\left\{4\right\}$};
	\node (F) [below left of=C, node distance=1.4cm, style={rectangle, draw}] {$\left\{1\right\}$};
	\node (G) [below right of=C, node distance=1.4cm, style={rectangle, draw}] {$\left\{3\right\}$};
	
	\draw[dashed](A) to node {}(B);
	\draw[](A) to node {}(C);
	\draw[dashed](B) to node {}(D);
	\draw[](B) to node {}(E);
	\draw[dashed](C) to node {}(F);
	\draw[](C) to node {}(G);
\end{tikzpicture}
}
\caption{A non-reducible determinization}\label{fig:mtbdd_ex_01_bad}
\end{figure}

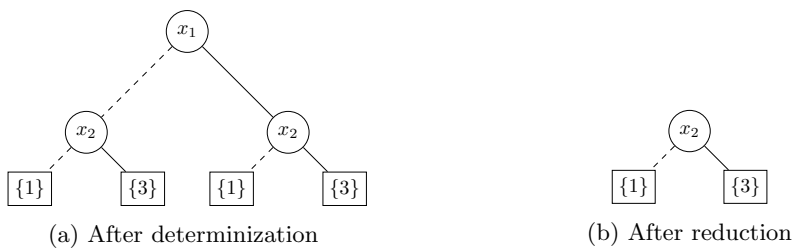
\begin{figure}
	\begin{subfigure}{.4\textwidth}
		\centering
		\scalebox{0.759}{
		\begin{tikzpicture}[>=stealth, node distance=2.5cm]
			\node (A) [style={circle, draw}] {$x_{1}$};
			\node (B) [below left of=A, style={circle, draw}] {$x_{2}$};
			\node (C) [below right of=A, style={circle, draw}] {$x_{2}$};
			\node (D) [below left of=B, node distance=1.4cm, style={rectangle, draw}] {$\left\{1\right\}$};
			\node (E) [below right of=B, node distance=1.4cm, style={rectangle, draw}] {$\left\{3\right\}$};
			\node (F) [below left of=C, node distance=1.4cm, style={rectangle, draw}] {$\left\{1\right\}$};
			\node (G) [below right of=C, node distance=1.4cm, style={rectangle, draw}] {$\left\{3\right\}$};
	
			\draw[dashed](A) to node {}(B);
			\draw[](A) to node {}(C);
			\draw[dashed](B) to node {}(D);
			\draw[](B) to node {}(E);
			\draw[dashed](C) to node {}(F);
			\draw[](C) to node {}(G);
		\end{tikzpicture}
		}
		\caption{After determinization}\label{fig:mtbdd_ex_01_good_a}
	\end{subfigure}
	\hfill
	\begin{subfigure}{.49\textwidth}
		\centering
		\vspace{1.3cm}
		\scalebox{0.75}{
		\begin{tikzpicture}[>=stealth, node distance=2.5cm]
			\node (B) [style={circle, draw}] {$x_{2}$};
			\node (D) [below left of=B, node distance=1.4cm, style={rectangle, draw}] {$\left\{1\right\}$};
			\node (E) [below right of=B, node distance=1.4cm, style={rectangle, draw}] {$\left\{3\right\}$};
	
			\draw[dashed](B) to node {}(D);
			\draw[](B) to node {}(E);
		\end{tikzpicture}
		}
		\caption{After reduction}\label{fig:mtbdd_ex_01_good_b}
	\end{subfigure}
	\caption{A reducible determinization}\label{fig:mtbdd_ex_01_good}
\end{figure}

In this paper, we suggest alternatives and hybrid approaches to overcome this potential shortcoming of \texttt{LA}, see Section~\ref{sec:det_algs}. Furthermore, to preserve information on safe initial states, we shall consider a modification of \texttt{LA} which forbids  assignment of ``any input'' to ``no-input'' grid cells.

\section{NP-completeness of determinization} \label{sec:np_com}
	\begin{theorem}
		The \texttt{OD} problem 
		is NP-complete (NP-C).
	\end{theorem}
	\begin{proof}
		To show that \texttt{OD} is NP-complete we prove that:
		
		\emph{(i)} \texttt{OD} is NP: Consider a non-deterministic algorithm\footnote{I.e. it has a non-deterministic step which always makes a ``proper'' (relative to the algorithm's goal) guess, see e.g.~\cite{Basu_2008}.}, Algorithm~\ref{alg:np_com:selection:od}, solving the selection \texttt{OD}. To prove that \texttt{OD} is NP, one must show that Algorithm~\ref{alg:np_com:selection:od} has polynomial time complexity and contains a polynomial number of random guesses, each independent on the problem size.  Clearly, all of the Algorithm~\ref{alg:np_com:selection:od} steps have polynomial time complexity. First of all, iterating over $Q \in 0, \ldots, \vert M \vert$,  checking for $\vert M'' \vert == Q$, and reducing $R\left(M'\right)$, see Section~\ref{sub:sec:prelim:bdds}, are polynomial time. Further, let us show that ``guessing $M'$ from $M$'' can also be done in polynomial time with a polynomial number of Bernoulli trials. 
		First, to visit $M$'s terminal nodes requires $\mathcal{O}\left(\vert M \vert\right)$ steps. Second, for a node, having at most $\vert \mathcal{U} \vert$ inputs, to randomly choose one input requires $\mathcal{O}\left( \ceil*{log_{2}\left(\vert \mathcal{U} \vert\right)}\right)$ Bernoulli trials. So we conclude that\footnote{All the missed auxiliary operations, e.g.: counting node inputs, removing the non-chosen inputs, and etc. are also polynomial time.} \texttt{OD} is NP.
		\begin{algorithm2e}
			\KwIn{$M$ - Controller's MTBDD}
			\KwResult{Size-optimal determinization of $M$}
			\For{$Q \in 0, \ldots, \vert M \vert$}{
				Guess $M'$ - a determinization of $M$\\
				$M'' \gets R\left(M'\right)$;\\
				\If{$\vert M'' \vert == Q$}
				{
					\Return{$M''$}
				}
			}
			\caption{NP algorithm for selection \texttt{OD}}
			\label{alg:np_com:selection:od}
		\end{algorithm2e}
		
		 \emph{(ii)} \texttt{MSC} is polynomial-time/space convertible to \texttt{OD}:
		 For an \texttt{MSC} with $X := \left\{ x_{i}\right\}^{N}_{i=1}$, $S := \left\{S_{j}\right\}^{K}_{j = 1}$, $\forall j \in \overline{1,K} : S_{j} \subseteq X$, and $N,K < \infty$, consider the next  three proving steps:
		 
		 \emph{a) Encode \texttt{MSC} as an MTBDD $M$}:
		 Take a binary tree with $N$ terminal nodes, indexed by $i \in \overline{1,N}$.
		 For each terminal node $i$ add a low (left) $t^{l}_{i}$ and a hight (right) $t^{h}_{i}$ children, such that $val\left(t^{l}_{i}\right) := \left\{ j \in \overline{1, K} \vert x_{i} \in S_{j} \right\}$ and $val\left(t^{h}_{i}\right) := \left\{ K + i \right\}$. Here, $val\left(.\right)$ is terminals' labelling function; the low terminals encode the \texttt{MSC} sets; and the high terminals prevent all but low-terminals' reductions. The resulting binary tree $T$ is a polynomial-space encoding of \texttt{MSC} as\footnote{Remember that we have $2\times N$ terminal nodes.}: $\vert T \vert  = 4\times N - 1$. Also, this is a polynomial-time encoding, as is realisable by Algorithm~\ref{alg:np_com:msc:od}, of time complexity $\mathcal{O}\left(N\right)$. To convert this binary tree into an MTBDD $M$, we shall interpret its non-terminal nodes as decision nodes, and its terminal nodes $t$ as value nodes, labeled with $val\left(t\right)$. This trivial conversion can also be done in polynomial time and space.

		 \emph{b) Solving \texttt{OD} for $M$ solves \texttt{MSC}}: 
		 For the given $M$ encoding of \texttt{MSC}, $R\left(M'\right)$, in Algorithm~\ref{alg:np_com:selection:od}, can only re-combine low terminals of $M'$ as high-terminal and thus non-terminal node reductions are prevented by the distinct high terminal node idexes. The high and non-terminals will stay intact and Algorithm~\ref{alg:np_com:selection:od} will effectively minimise the number of low terminals. The set $I$ of low-terminal labels of $M''$ then yields the solution for \texttt{MSC} as: \emph{(a)} $I$ defines a sub-cover of $S$; \emph{(b)} $\vert I \vert$ is minimal.
		 The former is clear as each $x_{i}$ of $X$ is related to a low-terminal. The latter can be proved by contradiction. First, fix low-terminal node values of $M$ to those of $I$ to get an MTBDD $M_{I}$ for which $\vert R\left(M_{I}\right) \vert = 3\times N - 1 + \vert I \vert$. Let us assume that $I$ is not a solution of \texttt{MSC} then there exists a sub-cover $I'$, such that $\vert I' \vert < \vert I \vert$. Similarly, for $M_{I'}$ we have $\vert R\left(M_{I'}\right) \vert = 3\times N - 1 + \vert I' \vert$, and thus $R\left(M_{I'}\right) < R\left(M_{I}\right)$. The latter contradicts the fact that Algorithm~\ref{alg:np_com:selection:od} solves \texttt{OD}.
		 
		 \emph{c) Decoding \texttt{MSC} solution from \texttt{OD} solution}:
		 Decoding of the \texttt{MSC} solution from $M''$ is straightforward: one needs to go through all of the low-terminal nodes and collect their labels. This requires linear time and space algorithm.
		 
		 Finally, since \texttt{MSC} is NP-C, \emph{(i)}\&\emph{(ii)}, imply that \texttt{OD} is NP-C.

		\begin{algorithm2e}
			\KwIn{$X = \left\{ x_{i}\right\}^{N}$ - Set elements}
			\KwIn{$S = \left\{S_{j}\right\}^{K}_{j = 1}$ - Set cover}
			\KwResult{The MTBDD encoding for \texttt{MSC}}
			$\mathrm{Node}\; root \gets \mathrm{null}$\\
			\If{$N > 0$}{
				$\mathrm{Queue\langle Node \rangle}\;terms\left(N\right) \gets \mathrm{empty}$\\
				$\mathrm{uint}\; cnt \gets \mathrm{1}$\\
			        $root \gets \mathbf{new}\; \mathrm{Node\left(\right)}$\\
			        $terms.push\_back\left(root\right)$\\
				\While{$cnt\; \not = N$}{
					$\mathrm{Node}\; term \gets terms.pop\_front\left(\right)$\\
					$terms.push\_back\left(term.low \gets \mathbf{new}\; \mathrm{Node\left(\right)}\right)$\\
					$terms.push\_back\left(term.high \gets \mathbf{new}\; \mathrm{Node\left(\right)}\right)$\\
					$tcnt \gets cnt + 1$
				}
				$\mathrm{uint}\; idx \gets \mathrm{1}$\\
				\While{$terms \not = \mathrm{empty}$}{
					$\mathrm{Node}\; term \gets terms.pop\_front\left(\right)$\\
					$term.low \gets \mathbf{new}\; \mathrm{Node\left(\right.}get\_set\_ids(x_{idx}, S\mathrm{\left.)\right)}$\\
					$term.high \gets \mathbf{new}\; \mathrm{Node\left(\right.}\left\{K + idx\right\}\mathrm{\left.\right)}$\\
					$idx \gets idx + 1$
				}
			}
			\Return{$root$};
			\caption{Encode \texttt{MSC} as \texttt{OD}}
			\label{alg:np_com:msc:od}
		\end{algorithm2e}
	\end{proof}

\section{Determinization algorithms} \label{sec:det_algs}
The newly suggested determinization algorithms have various underlying ideas: \texttt{GA} tries to maximise the number of states with the same input, and minimise the number of different inputs as a whole, both in an attempt to maximise the chances for (MT)BDD reductions; \texttt{LGA} combines complementary ideas of \texttt{LA} and \texttt{GA} to reduce the number of non-deterministic choices to be taken in the former one; \texttt{SR} attempts to find an analytical expression fitting the controller points on the largest part of its domain to reduce the number of distinct control mode areas to be stored;

	\subsection{Global Approach} \label{sub:sec:det_algs:ga}
		The \texttt{GA} approach is summarised in Algorithm~\ref{alg:det_algs:ga}, where: 
		\begin{enumerate}
		\item[\emph{(i)}] $inputs\_to\_states\left(.\right)$ creates $C$ -- the set of domain state indexes, $I$ -- the set of input indexes, and $\left\{C_{j}\right\}_{j \in I}$ -- the sets of states for the given inputs; 
		\item[\emph{(ii)}] $get\_min\_set\_cover(.)$ implements the \texttt{MSC} solution algorithm for unit set weights\footnote{The function returns a vector of inputs, ordered in the same way they were added to $I^{*}$, with the more common inputs coming first.}, see Section~\ref{sub:sec:prelim:min_set_cov}; 
		\item[\emph{(iii)}] $determinize\left(.\right)$ iterates over $I^{*}$ and for each state with the input removes all other inputs. 
		\end{enumerate}
		
		\begin{algorithm2e}
			\KwIn{$M$ - the controller's MTBDD}
			\KwResult{The determinized MTBDD}
			$\left(C, I, \left\{C_{j}\right\}_{j \in I}\right) \gets inputs\_to\_states(M)$\\
			$\mathrm{Vector}\; I^{*} \gets get\_min\_set\_cover\left(C, I, \left\{C_{j}\right\}_{j \in I}\right)$\\
			$M' \gets determinize\left(M, I^{*}\right)$\\
			\Return{$R\left(M'\right)$};
			\caption{The \texttt{GA} approach}
			\label{alg:det_algs:ga}
		\end{algorithm2e}

		\texttt{GA} differs from \texttt{LA} by looking at the state-space globally regardless of its' elements location. It maximizes the number of terminal nodes with identical labels, generally leading to a reduction in the number of used labels, which should facilitate MTBDD reductions.
	
	\subsection{Local-Global Approach} \label{sub:sec:det_algs:lga}
		Recall the MTBDD-based \texttt{LA} algorithm discussed in Section~\ref{sec:la_weak}. We showed that such determinization procedure can suffer from sub-optimal non-deterministic resolutions when multiple input choices are available in some regions. \texttt{LGA} combines \texttt{LA} with \texttt{GA} in an attempt to improve the resulting reductions by minimising this uncertainty. In essence, the \texttt{LGA} approach proceeds as \texttt{LA} up to the moment a non-trivial set of inputs, common for a grid area, is found; then the input is chosen according to the priority-descending order of inputs, as pre-computed by the $get\_min\_set\_cover\left(.\right)$ function, see Algorithm~\ref{alg:det_algs:ga}.

	\subsection{BDD-index based Local-Global Approach} \label{sub:sec:det_algs:blga}
		RO-(MT)BDDs achieve significant size reductions only if a ``good'' variable ordering is found, see Section~\ref{sub:sec:prelim:bdds}. Given the (MT)BDD encoding, see Equation~\ref{eq:cell_enc_cudd} of Section~\ref{sec:problem}, the variable reordering swaps grid-cell index bits realising a limited\footnote{Swapping bits affects all indexes; bits can not change value.} form of cell re-indexing.  Figure~\ref{fig:var_reo_ex} shows the effects thereof on the $\cl \subset \Nnn \times \Nnn$ function for an \texttt{LGA}-determinized BDD controller of the DC motor case study, see Section~\ref{sub:sec:empirical:case_studies}. The horizontal and vertical axes of the plots correspond to the state- and input-space element indexes respectively. The distinct vertical lines on Figures~\ref{fig:var_reo_ex:scots}~and~\ref{fig:var_reo_ex:bdd} are the $1.000$ point marks. According to Section~\ref{sec:problem}, the BDD's range of $\mathcal{S}$ indexes is wider than that of SCOTSv2.0.
		Comparing $\clf{.}$ in RO-BDD and SCOTSv2.0 indexes, the former exhibits better data clustering. To use this to our benefit, we suggest a version of \texttt{LGA}, called \texttt{BLGA}, using the RO-BDD indexes.

\begin{figure}
	\begin{subfigure}{.48\textwidth}
		\centering
		  \includegraphics[width=5.5cm, height=2cm]{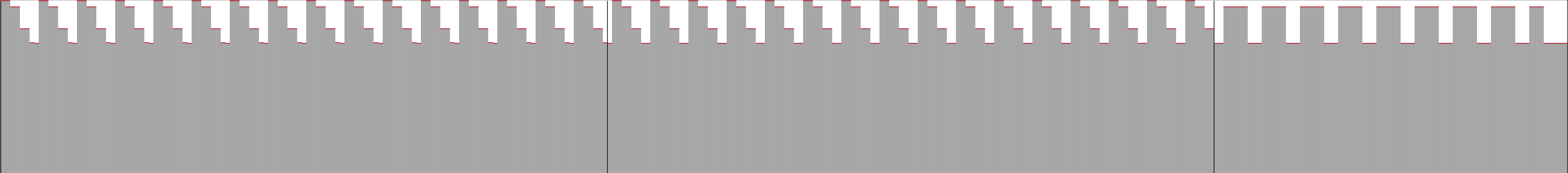}
		\caption{SCOTSv2.0 grid-cell index encoding}\label{fig:var_reo_ex:scots}
	\end{subfigure}
	\hfill
	\begin{subfigure}{.48\textwidth}
		\centering
		  \includegraphics[width=5.5cm, height=2cm]{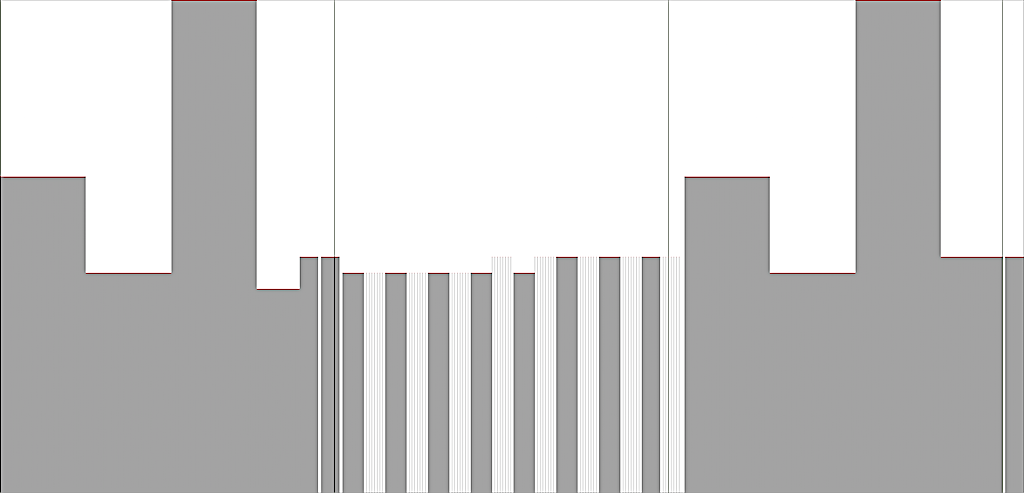}
		\caption{RO-BDD grid-cell index encoding}\label{fig:var_reo_ex:bdd}
	\end{subfigure}
	\caption{DC motor: controller function}\label{fig:var_reo_ex}
\end{figure}

	\subsection{Symbolic Regression} \label{sub:sec:det_algs:sr}
For the \texttt{SR} algorithm, a set of candidate controllers is evolved using a combination of GGGP and sep-CMA-ES, \emph{c.f.} references in Section~\ref{sub:sec:prelim:symb_reg}, using $i_{\max}$ individuals (i.e. candidate solutions) for $N$ generations. GGGP is used to evolve the functional structure of the controller based on a grammar and sep-CMA-ES to optimize the parameters. 
Given a candidate controller $\cl_{\texttt{SR}}: \R^n\rightarrow \R^m$, the fitness function $F$ with respect to a finite set $S$ 
is defined as: 
\begin{align*}
F(\cl_{\texttt{SR}},\mathcal{S}_c) =& \frac{\left| \left\{  s \in \mathcal{S}_c \mid  \texttt{q}_u(\cl_{\texttt{SR}}(s)) \in g(s) \right\}\right|}{|\mathcal{S}_c|}. \\
\end{align*}
In order to reduce the computation time, the set of states $\mathcal{S}_c$ is down-sampled to a set with a maximum of $\lambda$ elements. The reproduction involves selecting individuals based on tournament selection and the genetic operators crossover and mutation, in which parts of the individuals are exchanged or randomly altered respectively. More in depth descriptions of the used GGGP and sep-CMA-ES algorithms can be found in \cite{Verdier_EtAl_2017} and \cite{Ros_2008} respectively. After a maximum amount of generations the individual with the highest fitness is selected. For the resulting controller, it is verified for which states $s\in \mathcal{S}_c $ it holds that $\texttt{q}_u(\cl_{\texttt{SR}}(s)) \in g(s)$. For the remaining states the inputs are determinized using \texttt{GA}, \texttt{LA} or \texttt{LGA}. Finally, all states and corresponding new input indexes are again stored in a BDD. 

\newcommand{\rmm}[1]{\left< \mathrm{#1} \right>}

For our experiments we used the parameters in Table~\ref{tab:par}. Table~\ref{tab:PRexamp} shows the grammar employed to determine in which space analytical expressions to fit the controllers were selected. Here $\mathrm{sgn}$ denotes the signum function and Random Real $\in \left[-1,1\right]$ creates a random real within the specified interval.  The used starting symbol is $\rmm{strt}$. 

\begin{table}[ht]
\centering
\caption{\texttt{SR} parameters}
\label{tab:par}
    \begin{tabular}{rcl}
    parameters & value & explanation \\ \hline
 $\lambda$ & 1000 & Cardinality of down-sampled $\mathcal{S}_c$ \\
 $N$ & 50 &Max. number of GGGP generations. \\
 $M$ & 32 & Maximum number of individuals.\\
$d$ & 7 & Maximum tree depth of genotypes.\\
$(c_c, c_m)$ & $(0.5, 0.5)$ & Crossover and mutation chance. \\
$\sigma_0$ & 1& Initial standard deviation. \\
 $N_\mathrm{ES}$ &10 & Max. number of CMA-ES generations. \\
  \hline
    \end{tabular}
\end{table}

\begin{table}[ht]
\centering
\caption{Production rules $\mathcal{P}$ }
\label{tab:PRexamp}
    \begin{tabular}{rl}
    Nonterminal & Rules \\ \hline
    $\left< \mathrm{strt} \right> ::=$ & $\rmm{const} +\rmm{expr}  \mid \rmm{const}\cdot\rmm{expr}$ \\
    & $  \mid \rmm{const} +\rmm{const}\cdot\rmm{expr} $\\
$\left< \mathrm{expr} \right> ::=$ & $\rmm{lin} \mid \rmm{pol} \mid 0.5 \mathrm{sgn}( \rmm{lin}) + 0.5+\rmm{expr}$ \\ 
& $ \mid 0.5 \mathrm{sgn}( \rmm{pol}) + 0.5+\rmm{expr}$\\
$\left< \mathrm{lin} \right> ::=$ & $\rmm{const}\cdot x_1 + \dots \rmm{const}\cdot x_n$ \\
$\left< \mathrm{pol} \right> ::=$ & $\rmm{pol} + \rmm{pol} ~|~ \rmm{const}\times\rmm{mon}  $\\
$\left< \mathrm{mon} \right> ::=$ &$ \rmm{var} ~|~ \rmm{var} \times \rmm{mon} $ \\
 $\left< \mathrm{var} \right> ::=$ & $x_1 ~|~ \dots \mid x_n$ \\
  $\left< \mathrm{const} \right> ::=$ & Random Real $\in \left[-1,1\right]$ \\
  \hline
    \end{tabular}
\end{table}

\section{Empirical evaluation} \label{sec:empirical}

	\subsection{Case studies} \label{sub:sec:empirical:case_studies}
	All of the considered case-studies, but the last one, are taken from the standard distribution of SCOTSv2.0: \texttt{Aircraft} - a DC9-30 aircraft landing maneuver, see~\cite{Reissig_EtAl_2016}; \texttt{Vehicle} - a path planning problem for an autonomous vehicle, see~\cite{Zamani_EtAl_2012}~and~\cite{Reissig_EtAl_2016}; \texttt{DCDC} - a boost DC-DC converter with a reach-and-stay 
voltage specification, see~\cite{Girard_2012.1}; \texttt{DCDC rec 1/2} - the same as \texttt{DCDC} but enforces a recurrence specification for two targets; \texttt{DCM} - a DC motor with a reach-and-stay velocity specification, see~\cite{Mazo_EtAl_2010}.The symbolic BDD controller sizes were varied by modifying the models' input-/state-space discretisation parameters.

	\subsection{Software details} \label{sub:sec:empirical:soft}
	For the evaluation, we have realised the following software:
	\begin{itemize}
		\item A C++11 based LibLink library\footnote{We preferred LibLink over WSTP due to faster data-exchange.} for Mathematica~$11$, see~\cite{Mathematica_2017}, allowing to load and store BDD-based symbolic controllers of SCOTSv2.0.
		\item A C++11 based application implementing \texttt{LA}, \texttt{GA}, \texttt{LGA}, and \texttt{BLGA}. Our code is single-threaded as constrained by CUDD.
		\item A Mathematica~$11$ package implementing the \texttt{SR} approach. This realisation is natively multi-threaded and allows for a best utilisation of the CPU cores.
	\end{itemize}

	\subsection{Experimental setup} \label{sub:sec:empirical:set_up}
		We have measured: \emph{(i)}  determinization run-time in seconds as reported by the tools; \emph{(ii)} size of the determinized controllers in bytes, when stored to the file system. \texttt{SR} is probabilistic and therefore each of its experiments was repeated $5$ times. All other approaches are deterministic and thus their experiments were repeated only once.  Overall, we have considered the algorithms on various size models, varying the discretization parameters, and thus changing:
		\begin{enumerate}
			\item The number of model inputs:
				\begin{enumerate}
					\item \texttt{GA}, \texttt{LA}, \texttt{LGA}, \texttt{BLGA}
					\item \texttt{SR} on \texttt{LGA} determinized controllers
				\end{enumerate}
			\item The number of model states:
				\begin{enumerate}
					\item \texttt{GA}, \texttt{LA}, \texttt{LGA}, \texttt{BLGA}
				\end{enumerate}
		\end{enumerate}
		\texttt{SR} was only done on \texttt{LGA} determinized controllers because it: \emph{(i)} did not scale well with the growing number of inputs; \emph{(ii)} if feasible, shall be capable of reducing deterministic controllers.
		The experiments were done on two machines: \emph{(A)} MacBook Pro with: Intel i5 CPU (4 cores) $2.9$ GHz; $8$ GB $2133$ Mhz RAM; MacOS v$10.12.6$; \emph{(B)} PC with: Intel Xeon CPU ($8$ cores) E$5-1660$ v$3$ $3.00$GHz; $32$ GB $2133$ MHz RAM; Ubuntu 16.04.3 LTS. The type $(1.a)$ experiments ran on machine \emph{(A)}; $(1.b)$ and $(2)$ on $(B)$.
		
		Given, a significant difference in software realization (Mathematica v.s. C++11, multi v.s. single threaded), running \texttt{SR} on faster multicore machine, and that controllers' determinization is an offline job, our run-time data: \emph{(i)} is only dedicated to show the approaches' feasibility; \emph{(ii)} can only hint the actual performance differences between \texttt{SR} and others. This is why also the run-time for  \texttt{LA}, \texttt{GA}, \texttt{LGA}, and \texttt{BLGA} is not averaged over multiple re-runs.

	\subsection{Results} \label{sub:sec:empirical:multi_input}
	Table~\ref{tbl:multi_input:main} presents the core experimental data for models obtained by varying the number of inputs. Here, column: ``$\mbox{SCOTS}$'' lists information for the original controllers; ``$\mbox{Time}$'' is the algorithm's run-time in seconds; ``$\texttt{A-SR}$'' and ``$\texttt{M-SR}$''stand for the average and maximum \texttt{SR} values over $5$ repetitions; and ``$\mbox{Fit \%}$'' is the fitness percentage of the \texttt{SR} controller's symbolic part.
	To compare the compressing power of the approaches, for an algorithm $\omega$ and a case study $\nu$ we define size compression as: $C^{\nu}_{\omega}$$:=$$\left(1 - \vert B^{\nu}_{\omega} \vert/\vert B^{\nu}\vert\right)*100$, where $B^{\nu}$ and $B^{\nu}_{\omega}$ stand for the original and $\omega$-determinized BDD sizes. Comparing algorithms ``$\omega_{1}$ v.s. $\omega_{2}$'' is done by computing a difference $\Delta^{\nu}_{\omega_{1},\omega_{2}}$$:=$$C^{\nu}_{\omega_{1}}-C^{\nu}_{\omega_{2}}$. Clearly,  $\Delta^{\nu}_{\omega_{1},\omega_{2}} > 0 $ means $\omega_{1}$ being better than $\omega_{2}$ on $\nu$. Taking into account the \texttt{A-SR} experiment repetitions, we define\footnote{The mean value over $5$ experiment repetitions of \texttt{SR} on $\nu$.}: $C^{\nu}_{\mbox{\tiny\texttt{A-SR}}} := E\left[C^{\nu}_{\mbox{\tiny\texttt{SR}}}\right]$.
	
	Figure~\ref{fig:multi_input:size} contains two compression comparison sets: \emph{(i)} \texttt{GA}, \texttt{LGA}, \texttt{BLGA} v.s. \texttt{LA}; and \emph{(ii)} \texttt{A-SR}, \texttt{M-SR} v.s. \texttt{LGA}\footnote{Since \texttt{SR} was applied to the \texttt{LGA}-determinized controllers.}. The plot features mean compressions and the standard deviation thereof. We conclude the next compression ranking of the algorithms: \emph{1.} \texttt{LGA}, \emph{2.} \texttt{BLGA}, \emph{3.} \texttt{LA}, \emph{4.} \texttt{GA}, \emph{5.} \texttt{M-SR} \emph{6.} \texttt{A-SR}. 
	
	Figure~\ref{fig:multi_input:time} summarises the execution times for the set-up of Table~\ref{tbl:multi_input:main}. Relative to \texttt{LA}, on average: \texttt{GA} is $\approx0.8$ times faster;  \texttt{LGA} is comparable;  \texttt{BLGA} is $\approx1.1$ times slower; \texttt{A-SR} is $\approx180$ times slower but has a huge deviation of $\approx174$. The latter is due to probabilistic nature of \texttt{SR}. Note that, \texttt{A-SR} is multi-threaded and was run on a faster machine than the single-threaded \texttt{LA}. So the actual performance difference between the algorithms is more significant.
	
	Additionally, we compared \texttt{GA}, \texttt{LGA}, \texttt{BLGA} and \texttt{LA} on up to $52$~Mb size BDD controllers, obtained by varying the number of system states. These experiments only strengthened the algorithms' ranking conclusions implied by Figure~\ref{fig:multi_input:size}. We omit further detail on that, to save space.
	
	To conclude, we present Figure~\ref{fig:all:lga_rel_la} summarising the compression of \texttt{LGA} relative to \texttt{LA} on all of the $67$ considered BDD controllers. Per case-study $\nu$ the compression is computed as: $C^{\nu}_{\mbox{\tiny \texttt{LGA}, \texttt{LA}}}$$:=$$\left(1 - \vert B^{\nu}_{\mbox{\tiny \texttt{LGA}}} \vert/\vert B^{\nu}_{\mbox{\tiny \texttt{LA}}}\vert\right)*100$. The plot on the left of Figure~\ref{fig:all:lga_rel_la} shows the discretized distribution of $C^{\nu}_{\mbox{\tiny \texttt{LGA}, \texttt{LA}}}$, the plot on the right shows its mean and standard deviation. Notice that, on average, \texttt{LGA} produces $\approx 14$\% smaller controllers than \texttt{LA}, in the best case  \texttt{LGA} was capable of delivering up to $\approx 85$\% smaller controllers.
	
	\begin{table}
	\caption{Core experimental data}\label{tbl:multi_input:main}
	\vspace{-0.5cm}
	\begin{center}
	{\footnotesize 
	\[
	\begin{array}{|c|c|c|c|c|c|c|c|}
	\hline
	 &  \multicolumn{2}{c|}{\mbox{\tiny SCOTS}} & \multicolumn{2}{c|}{\mbox{\tiny LA}} &\mbox{\tiny LGA} &  \mbox{\tiny A-SR}  & \mbox{\tiny M-SR} \\ \hline
	 & \mbox{\tiny \#inputs} & \mbox{\tiny \#Bytes} & \mbox{\tiny \#Bytes} & \mbox{\tiny Time} & \mbox{\tiny \#Bytes} & \mbox{\tiny Time} & \mbox{\tiny Fit \%} \\ \hline
	 \rule{0pt}{7pt}\multirow{3}{*}{\rotatebox[origin=c]{90}{\mbox{Aircraft}}} & 20 & 2878481 & 150459 & 121.81 & 150316 & 1065,50 & 48.56 \\ \cline{2-8}
	 \rule{0pt}{7pt}& 57 & 9563407 & 193590 & 159.61 & 193055 & 1547,92 & 43.55 \\ \cline{2-8}
	 \rule{0pt}{7pt}& 112 & 8533274 & 236753 & 183.62 & 235273 & 1949,12 & 40.58 \\ \hline
	\multirow{5}{*}{\rotatebox[origin=c]{90}{\mbox{Vehicle}}} & 49 & 21972 & 10462 & 1.38 & 9821 & 572,77 & 32.31 \\ \cline{2-8}
	 & 169 & 28537 & 11956 & 1.79 & 11047 & 614,85 & 27.15 \\ \cline{2-8}
	 & 441 & 54692 & 19430 & 2.86 & 17357 & 770,81 & 13.96 \\ \cline{2-8}
	 & 729 & 52447 & 18435 & 3.41 & 15793 & 954,53 & 18.14 \\ \cline{2-8}
	 & 1087 & 60757 & 18939 & 4.04 & 16338 & 1455,87 & 24.82 \\ \hline
	\multirow{6}{*}{\rotatebox[origin=c]{90}{\mbox{DCM}}} & 2001 & 4951 & 830 & 2.04 & 371 & 458,61 & 33.14 \\ \cline{2-8}
	 & 10001 & 11957 & 1000 & 13.3 & 420 & 639,11 & 33.14 \\ \cline{2-8}
	 & 20001 & 24206 & 1166 & 34.65 & 410 & 742,10 & 24.43 \\ \cline{2-8}
	 & 30001 & 19161 & 1306 & 63.45 & 441 & 951,10 & 33.14 \\ \cline{2-8}
	 & 40001 & 13772 & 1308 & 94.00 & 449 & 975,50 & 19.82 \\ \cline{2-8}
	 & 50001 & 12921 & 1252 & 143.13 & 448 & 1121,98 & 33.14 \\ \hline
	\multirow{6}{*}{\rotatebox[origin=c]{90}{\mbox{DCDC}}} & 2 & 4532 & 786 & 0.75 & 786 & 431,15 & 94.19 \\ \cline{2-8}
	 & 45 & 5218 & 1025 & 1.57 & 1025 & 440,99 & 94.19 \\ \cline{2-8}
	 & 89 & 5350 & 1030 & 2.38 & 1030 & 351,11 & 94.18 \\ \cline{2-8}
	 & 134 & 5272 & 1036 & 3.31 & 1035 & 450,36 & 94.17 \\ \cline{2-8}
	 & 178 & 5266 & 1036 & 4.11 & 1035 & 368,17 & 94.15 \\ \cline{2-8}
	 & 223 & 5300 & 1037 & 5.09 & 1036 & 426,01 & 93.55 \\ \hline
	\multirow{6}{*}{\rotatebox[origin=c]{90}{\mbox{DCDC r1}}} & 2 & 4247 & 773 & 0.78 & 773 & 448,00 & 97.35 \\ \cline{2-8}
	 & 45 & 6009 & 915 & 1.48 & 915 & 417,70 & 97.35 \\ \cline{2-8}
	 & 89 & 5615 & 921 & 2.15 & 921 & 422,11 & 97.35 \\ \cline{2-8}
	 & 134 & 5768 & 936 & 2.96 & 930 & 359,65 & 97.35 \\ \cline{2-8}
	 & 178 & 5781 & 936 & 3.64 & 930 & 409,91 & 97.35 \\ \cline{2-8}
	 & 223 & 5714 & 936 & 4.48 & 930 & 428,49 & 95.13 \\ \hline
	\multirow{6}{*}{\rotatebox[origin=c]{90}{\mbox{DCDC r2}}} & 2 & 2243 & 791 & 0.73 & 828 & 439,04 & 95.08 \\ \cline{2-8}
	 & 45 & 3685 & 934 & 2.15 & 937 & 428,78 & 94.51 \\ \cline{2-8}
	 & 89 & 3638 & 939 & 3.66 & 943 & 395,83 & 95.12 \\ \cline{2-8}
	 & 134 & 3565 & 949 & 4.98 & 949 & 361,49 & 94.99 \\ \cline{2-8}
	 & 178 & 3531 & 949 & 6.43 & 949 & 456,36 & 94.82 \\ \cline{2-8}
	 & 223 & 3549 & 950 & 8.13 & 950 & 408.70 & 92.91 \\ \hline
	\end{array}
	\]
	}
	\end{center}
	\end{table}
	
	\begin{figure}
		\centering
		 \includegraphics[scale=0.98]{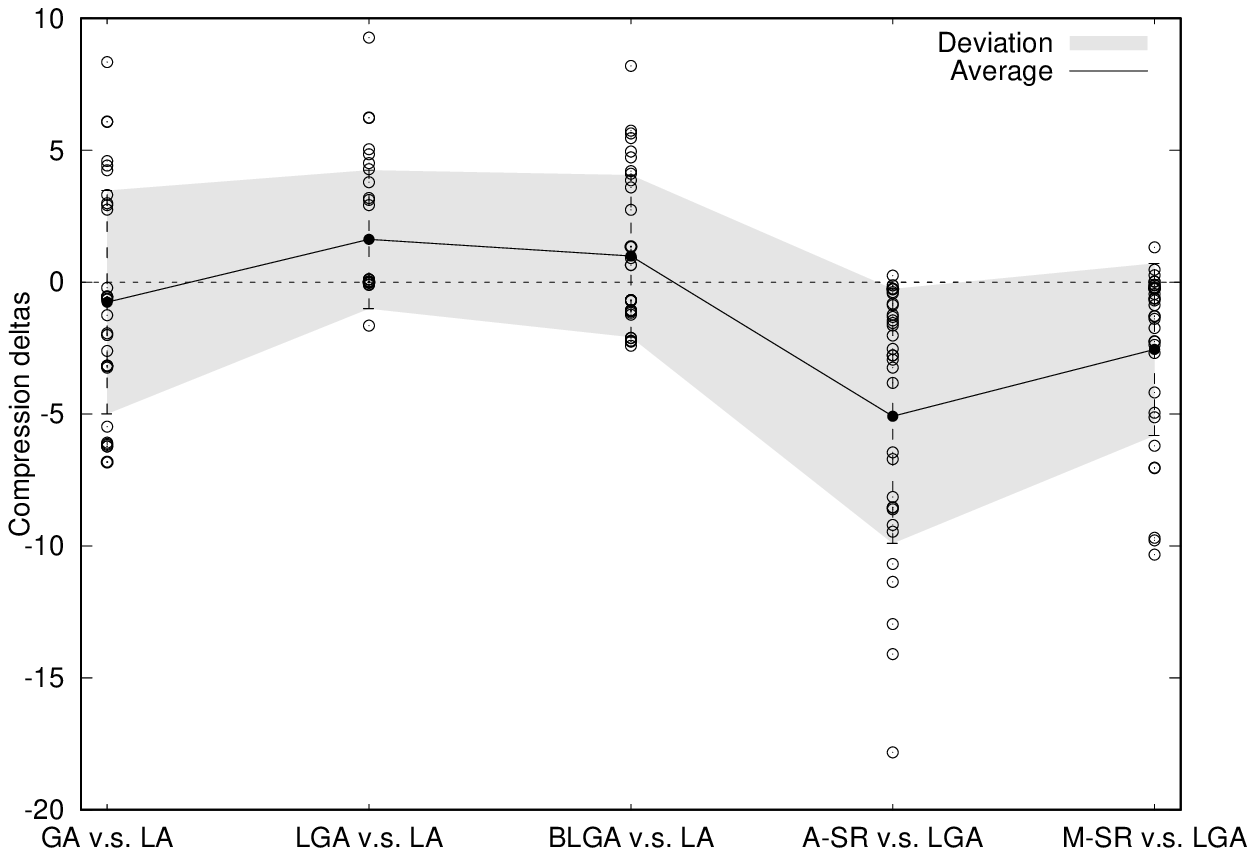}
		\caption{Plotting $\Delta^{\nu}_{\mbox{\tiny \texttt{GA}},\mbox{\tiny \texttt{LA}}}$, $\Delta^{\nu}_{\mbox{\tiny \texttt{LGA}},\mbox{\tiny \texttt{LA}}}$, $\Delta^{\nu}_{\mbox{\tiny \texttt{BLGA}},\mbox{\tiny \texttt{LA}}}$, $\Delta^{\nu}_{\mbox{\tiny \texttt{A-SR}}, \mbox{\tiny \texttt{LGA}}}$, $\Delta^{\nu}_{\mbox{\tiny \texttt{M-SR}},\mbox{\tiny \texttt{LGA}}}$}\label{fig:multi_input:size}
	\end{figure}

	\begin{figure}
		\centering
		\hspace{-0.4cm}
		\begin{tabular}{cc}
			\includegraphics[scale=0.8]{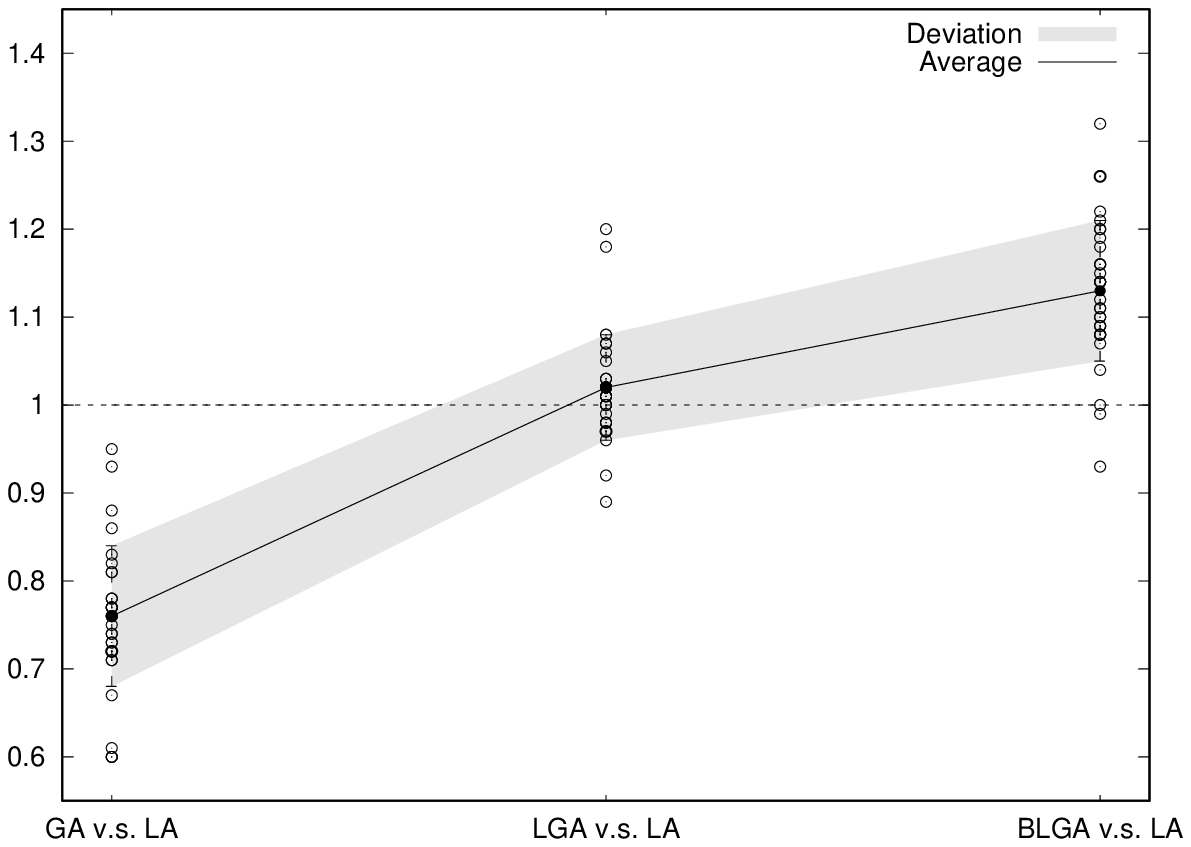} & \hspace{-0.7cm}
			\includegraphics[scale=0.8]{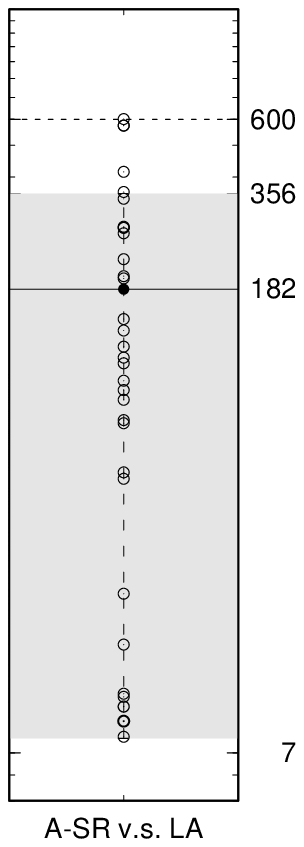}
		\end{tabular}
		\caption{Execution times ratio relative to \texttt{LA}}\label{fig:multi_input:time}
	\end{figure}
	
	\begin{figure}
		\centering
		\hspace{-0.8cm}
		\begin{tabular}{cc}
			\includegraphics[scale=0.75]{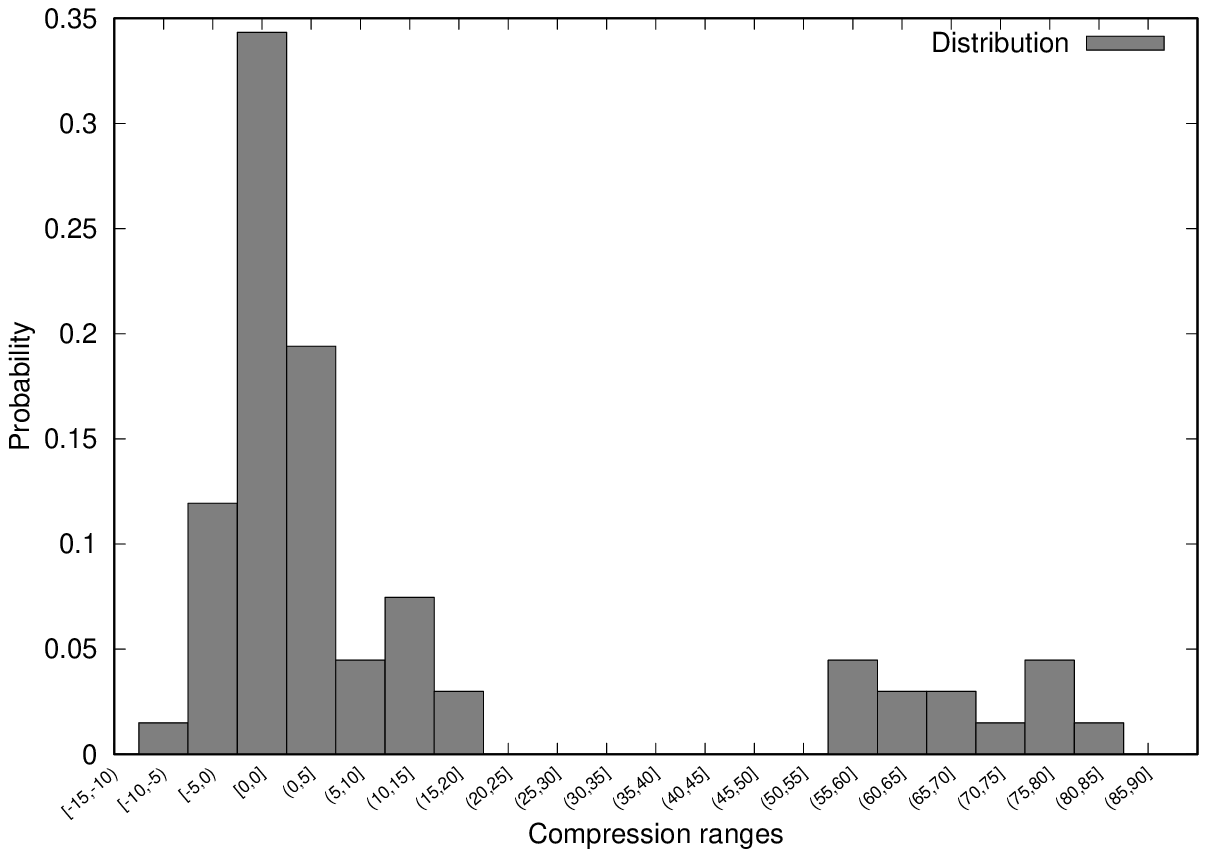} & \hspace{-0.7cm}
			\includegraphics[scale=0.75]{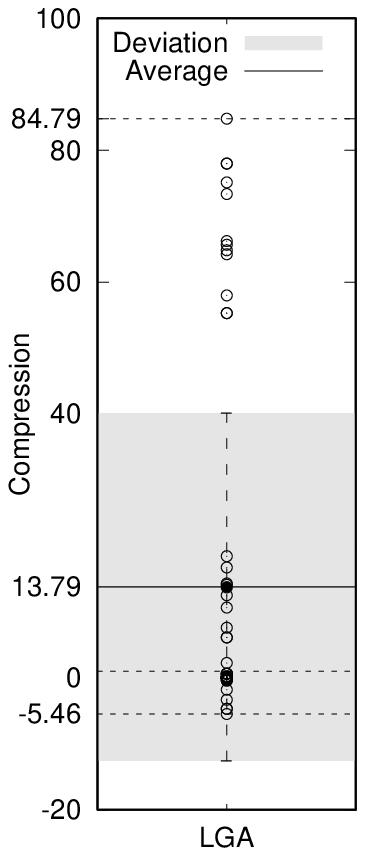}
		\end{tabular}
		\caption{$C^{\nu}_{\mbox{\tiny \texttt{LGA}, \texttt{LA}}}$ distribution, deviation and average}\label{fig:all:lga_rel_la}
	\end{figure}

\section{Conclusions} \label{sec:empirical:summ}

	In this work, we have considered the problem of size-optimal BDD controllers determinisation (\texttt{OD}), which we show to be NP-complete.
	Up until now, the only heuristic approach to solve \texttt{OD} was proposed by \cite{Girard_2012} and was based on representing the controller function as a binary tree. We have shown how that such an approach, which we call \texttt{LA}, can be extended to use the more size-efficient RO-(MT)BDDs data structure.
	In addition, we have identified examples where \texttt{LA} is sub-optimal due to only considering controller's local properties.
	A global approach (\texttt{GA}), based on the minimum set cover problem solution algorithm, was proposed to remedy this.
	Further, a hybrid of \texttt{GA} and \texttt{LA}, called \texttt{LGA}, was suggested to incorporate the strengths of both approaches.
	To exploit the clustering of internal BDD indexes, we have come up with a BDD-index based version of \texttt{LGA}, called \texttt{BLGA}.
	Finally, we made an attempt of substituting the BDD-based control-law representations by functions generated using the symbolic regression (genetic-algorithm powered) approach, we refer to as \texttt{SR}.
	
	All of the devised approaches were compared in compressing power and run-time by means of an extended empirical evaluation. 
	The compression ranking of the algorithms turns out to be: \emph{1.} \texttt{LGA}, \emph{2.} \texttt{BLGA}, \emph{3.} \texttt{LA}, \emph{4.} \texttt{GA}, \emph{5.} \texttt{SR}.
	The run times of \texttt{LA}, \texttt{GA}, \texttt{LGA}, and \texttt{BLGA} are of the same order but \texttt{SR} is at least one to two orders of magnitude slower.
	
	In principle, \texttt{SR} could allow us to eliminate BDDs completely, leading to potentially smaller functional expressions and prevent using BDD-data accessing code that, as for CUDD, is difficult (and size expensive) to port to embedded hardware.
	We did not manage to achieve that due to:
	\emph{(i)} our \texttt{SR} realization not being powerful enough, see low fitness values in Table~\ref{tbl:multi_input:main};
	\emph{(ii)} using BDDs for storing the controller's support, due to a decision to preserve controller's domain.
	For now, we shall note that \texttt{SR} still looks promising for getting small and practical controllers. However, symbolic controllers seem to have structure that is not easy for \texttt{SR} to achieve a $100$\% fitness on. So more research is needed to be done in this direction.

\bibliographystyle{plain}
\bibliography{cclr-arxiv}

\end{document}